\documentclass[10pt,twocolumn]{article}

\usepackage[utf8]{inputenc}
\usepackage{amsmath, amssymb, amsfonts, amsthm} 
\usepackage{titlesec}
\usepackage{cite}
\usepackage{booktabs}
\usepackage{graphicx}
\usepackage[colorlinks=false]{hyperref}
\usepackage[format=plain,labelfont=it]{caption}
\usepackage[left=1.5cm,right=1.5cm,top=2cm,bottom=2cm]{geometry}
\usepackage{color}
\usepackage{fancyhdr}
\usepackage{xurl}
\usepackage{svg}

\titleformat{\section}{\centering\normalfont\scshape}{\Roman{section}.}{5pt}{}
\titleformat{\subsection}{\normalfont\it}{\Alph{subsection}.}{5pt}{}
\titleformat{\subsubsection}{\normalfont\it}{\hspace{4mm}\arabic{subsubsection})}{5pt}{}

\newcommand\infoFootnote[1]{%
  \begingroup
  \renewcommand\thefootnote{}\footnote{#1}%
  \addtocounter{footnote}{-1}%
  \endgroup}

\newtheorem{thm}{Theorem}

\newtheorem{lem}[thm]{Lemma}

\newtheorem{assum}{Assumption}

\newtheorem{defn}{Definition}

\newtheorem{rem}{Remark}

\newcommand{\R}{\mathbb{R}} 

\newcommand{\ab}{\boldsymbol{a}}
\newcommand{\fb}{\boldsymbol{f}}
\newcommand{\gb}{\boldsymbol{g}}
\newcommand{\sbb}{\boldsymbol{s}}
\newcommand{\ub}{\boldsymbol{u}}
\newcommand{\wb}{\boldsymbol{w}}
\newcommand{\xb}{\boldsymbol{x}}
\newcommand{\yb}{\boldsymbol{y}}
\newcommand{\zb}{\boldsymbol{z}}
\newcommand{\xib}{\boldsymbol{\xi}}
\newcommand{\alphab}{\boldsymbol{\alpha}}
\newcommand{\betab}{\boldsymbol{\beta}}
\newcommand{\zerob}{\boldsymbol{0}}
\newcommand{\oneb}{\boldsymbol{1}}

\newcommand{\ybs}{\mathbf{y}}
\newcommand{\xbs}{\mathbf{x}}
\newcommand{\ubs}{\mathbf{u}}

\newcommand{\Ab}{\boldsymbol{A}}
\newcommand{\Fb}{\boldsymbol{F}}
\newcommand{\Hb}{\boldsymbol{H}}
\newcommand{\Ib}{\boldsymbol{I}}
\newcommand{\Ub}{\boldsymbol{U}}
\newcommand{\Wb}{\boldsymbol{W}}
\newcommand{\Xb}{\boldsymbol{X}}
\newcommand{\Yb}{\boldsymbol{Y}}
\newcommand{\Qb}{\boldsymbol{Q}}
\newcommand{\Eb}{\boldsymbol{E}}
\newcommand{\Gb}{\boldsymbol{G}}
\newcommand{\Vb}{\boldsymbol{V}}
\newcommand{\Zb}{\boldsymbol{Z}}

\newcommand{\Cbc}{\boldsymbol{\mathcal{C}}}
\newcommand{\Dbc}{\boldsymbol{\mathcal{D}}}
\newcommand{\Qbc}{\boldsymbol{\mathcal{Q}}}
\newcommand{\Wbc}{\boldsymbol{\mathcal{W}}}

\newcommand{\Nc}{\mathcal{N}}
\newcommand{\Uc}{\mathcal{U}}
\newcommand{\Yc}{\mathcal{Y}}

\newcommand{\Dfrak}{\mathfrak{D}}

\newcommand{\image}[1]{\mathcal{R}\left({#1}\right)}
\newcommand{\rank}[1]{\mathrm{rank}\left({#1}\right)}
\newcommand{\kernel}[1]{\mathcal{N}\left({#1}\right)}
\newcommand{\conv}[1]{\mathrm{conv}\left({#1}\right)}
\newcommand{\blkdiag}[1]{\mathrm{diag}\left({#1}\right)}

\DeclareMathOperator*{\argmin}{\arg\min}

\renewcommand{\boldsymbol}[1]{#1}
\renewcommand{\mathbf}[1]{\mathrm{#1}}

\def\tvdots{\vbox{\baselineskip=2pt \lineskiplimit=0pt \kern6pt \hbox{.}\hbox{.}\hbox{.}}}

\title{\vspace{-2mm}\bf On data usage and predictive behavior of data-driven predictive control \\ with 1-norm regularization}
\author{Manuel Kl\"adtke and Moritz Schulze Darup \vspace{2mm}}
\date{}

\fancypagestyle{plain}{%
  \renewcommand{\headrulewidth}{0pt}%
  \fancyhf{}%
  \fancyfoot[C]{\footnotesize 2475-1456 \copyright 2025 IEEE}%
}

\begin{document}
\maketitle

\pagestyle{fancy}
\fancyhf{} 
\renewcommand{\headrulewidth}{0pt}
\fancyfoot[C]{\footnotesize 2475-1456 \copyright 2025 IEEE}

\begin{abstract}
We investigate the data usage and predictive behavior of data-driven predictive control (DPC) with 1-norm regularization. Our analysis enables the offline removal of unused data and facilitates a comparison between the identified symmetric structure and data usage against prior knowledge of the true system. This comparison helps assess the suitability of the DPC scheme for effective control.
\end{abstract}
\infoFootnote{M. Kl\"adtke and M. Schulze Darup are with the \href{https://rcs.mb.tu-dortmund.de/}{Control and~Cyberphysical Systems Group}, Faculty of Mechanical Engineering, TU Dortmund University, Germany. E-mails:  \href{mailto:manuel.klaedtke@tu-dortmund.de}{\{manuel.klaedtke, moritz.schulzedarup\}@tu-dortmund.de}. \vspace{0.5mm}}
\infoFootnote{This paper is a \textbf{preprint} of a contribution to the IEEE Control Systems Letters. The DOI of the original paper is \href{https://doi.org/10.1109/LCSYS.2025.3575436}{10.1109/LCSYS.2025.3575436}.} 

\section{Introduction}
\label{sec:Introduction}

Data-driven predictive control (DPC) \cite{Yang2013, Coulson2019DeePC} is a popular alternative to model predictive control (MPC) using linear combinations of trajectory data instead of a model for predictions. Utilizing Willems' Fundamental Lemma \cite{WILLEMS2005}, DPC can provide exact predictions for deterministic LTI systems or certain classes of nonlinear systems \cite{Berberich2020}. However, noise, disturbances, and general nonlinearities typically compromise this exactness. A common solution is to add a regularization term, justified through distributional robustness \cite{Coulson2019RegularizedDeePC} or as convex relaxations of model-based predictive control schemes \cite{Dorfler2021}.
Popular regularizations include quadratic \cite{Berberich2020stability} and 1-norm \cite{Coulson2019DeePC} approaches. 
The 1-norm regularization considered in this paper serves as a convex relaxation of low-rank approximation \cite{MARKOVSKY2008surveySLRA} and thus ``controls the model complexity'' \cite{Dorfler2021} by limiting the number of data trajectories used for predictions.
Hence, in nonlinear DPC, it is expected to choose data ``that is most relevant to the current operating point'' \cite{Mattsson2021}.
However, a specific analysis of its data usage is currently missing in the literature. 

Motivated by this gap, this work characterizes DPC's data usage and predictive behavior using previously established analysis tools \cite{KLAEDTKE2023, Klaedtke2025_AT_Accepted}. We demonstrate that some data is never utilized, connect the trajectory-specific cost of 1-norm regularization to the atomic norm \cite{Chandrasekaran2012_atomicNorm}, and show that the resulting predictive behavior forms a symmetric piecewise affine (PWA) function with subdomains dictating local data usage that scale linearly with the regularization weight. The presentation of these key findings is organized as follows. Section~\ref{sec:fundamentals} summarizes key preliminaries on DPC and introduces the analysis tools employed. Section~\ref{sec:main} contains our core results, analyzing data usage, trajectory-specific effects of 1-norm regularization, and implicit predictive behavior. Finally, Section~\ref{sec:Conclusions} concludes our findings and previews future extensions of this analysis.

\section{Fundamentals}\label{sec:fundamentals}

\subsection{Fundamentals of direct data-driven predictions}
Instead of utilizing a discrete-time state-space model like
$$
    \xb(k+1)=\fb(\xb(k), \ub(k)), \qquad \quad \yb(k) = \gb(\xb(k), \ub(k))
$$
with input $\ub\in\R^m$, state $\xb\in\R^n$, and output ${\yb\in\R^p}$ as in MPC, predictions in DPC are realized based on previously collected trajectory data $\big(\ubs^{(1)}, \ybs^{(1)}\big), \hdots, \big(\ubs^{(\ell)}, \ybs^{(\ell)}\big)\in\R^{m L}\times \R^{p L}$ via linear combinations
\begin{equation}
    \begin{pmatrix}
        \ubs_\text{gen}\\
        \ybs_\text{gen}
    \end{pmatrix}
    =
    \begin{pmatrix}
        \ubs^{(1)}\\
        \ybs^{(1)}
    \end{pmatrix} a_1 + \hdots + 
    \begin{pmatrix}
        \ubs^{(\ell)}\\
        \ybs^{(\ell)}
    \end{pmatrix} a_\ell
    = 
    \Dbc
    \ab. \label{eq:linearComb}
\end{equation}
 Here, the dimensions of the data matrix $\Dbc\in\R^{L(m+p)\times\ell}$ and generator vector $\ab := \begin{pmatrix}
    \ab_1 & \hdots & \ab_\ell
\end{pmatrix}^\top\in\R^\ell$ are specified by the length~$L$ of recorded (and generated) trajectories and the number~$\ell$ of data trajectories used for predictions. 
Note that we use upright notation to denote sequences of system quantities, e.g., the $i$-th input sequence in $\Dbc$ being $\ubs^{(i)} = \begin{pmatrix}
    \ub(k_i)^\top & \hdots & \ub(k_i+L-1)^\top
\end{pmatrix}^\top$ for some $k_i$.
Assuming exact data generated by an LTI system and that $L$ exceeds the system's lag, the image $\image{\Dbc}$ is equivalent to the set of all possible system trajectories (of length $L$) if and only if \cite{Markovsky2020}
\begin{equation}\label{eq:GPE} 
    \rank{\Dbc}=L m + n. 
\end{equation} 
The generalized persistency of excitation condition \eqref{eq:GPE} justifies the linear combinations \eqref{eq:linearComb}, since 
\begin{equation}
    \begin{pmatrix}
        \ubs_\text{gen}\\
        \ybs_\text{gen}
    \end{pmatrix}
    \in
    \image{\Dbc} \; \iff \; \exists \ab\; \text{such that}\; 
    \begin{pmatrix}
        \ubs_\text{gen}\\
        \ybs_\text{gen}
    \end{pmatrix} = \Dbc\ab. \label{eq:imageRep0}
\end{equation}
Representing system trajectories in this way is also known as an image representation, in contrast to, e.g., a state-space representation. A popular sufficient condition for data to satisfy \eqref{eq:GPE} is known as Willems' Fundamental Lemma \cite{WILLEMS2005}, which has become synonymous with using image representations. 
To include the current initial condition of the system as a starting point for predicted trajectories, the generated I/O-sequence is typically partitioned into a past section $(\ubs_p, \ybs_p)$ and a future section $(\ubs, \ybs)$ with $N_p$ respectively $N$ time-steps yielding
$$
    \begin{pmatrix}
        \ubs_p\\
        \ubs
    \end{pmatrix} = \ubs_\text{gen} 
    = 
    \begin{pmatrix}
        \Ub_p\\
        \Ub
    \end{pmatrix} \ab \quad \text{and} \quad
    \begin{pmatrix}
        \ybs_p\\
        \ybs
    \end{pmatrix} = \ybs_\text{gen} 
    = 
    \begin{pmatrix}
        \Yb_p\\
        \Yb
    \end{pmatrix} \ab. 
$$
The past section of a predicted trajectory is then forced to match the I/O-data $\xib$ recorded in the most recent $N_p$ time-steps during closed-loop operation, i.e., the constraints
$$
    \xib= 
    \begin{pmatrix}
        \ubs_p \\ \ybs_p
    \end{pmatrix}
    =
    \begin{pmatrix}
        \Ub_p \\ \Yb_p
    \end{pmatrix}\ab
    = \Wb \ab
$$
force any predicted trajectory to start with the most recently witnessed behavior of the system, acting as an initial condition. Note that $\xib$ is indeed a (non-minimal) state of the LTI system, if $N_p$ is chosen greater or equal to its lag. 
For more concise notation, we define 
$\Zb := \begin{pmatrix}
        \Wb^\top & \Ub^\top
    \end{pmatrix}\!^\top,\,\,  \zb := \begin{pmatrix}
        \xib^\top & \ubs^\top
    \end{pmatrix}\!^\top,\,\,   
    \wb := \begin{pmatrix}
        \zb^\top & \ybs^\top
    \end{pmatrix}\!^\top\!.$
    Additionally, with a slight abuse of notation, we redefine
    $\Dbc := \big(
        \Zb^\top \quad \Yb^\top
\big)^\top,$
which is a block-row permutation to reflect this new partitioning.
\begin{rem}
    \label{rem:Statespace}
    Although we have introduced the data-driven predictions in an I/O setting, they can be straightforwardly modified to a state-space setting \cite{DePersis2020}.  To this end, consider the initial state $\xb_0\in\R^n$ and predicted state sequence $\xbs\in \R^{n N}$, which replace $\xib$ and $\ybs$. The corresponding data matrices are $\Xb_0\in\R^{n\times\ell}$ and $\Xb\in\R^{n N\times \ell}$, replacing $\Wb$ and $\Yb$.
    This leads to a data matrix $\Dbc\in \R^{n+(m+n)N\times \ell}$.
    Our considerations and condition \eqref{eq:GPE} with $\rank{\Dbc}=n+m N$ hold for both settings, and we employ     the latter for visualization of low dimensional examples in Fig.~\ref{fig:scalingCRs} and~\ref{fig:symmetryPredictor}.
\end{rem}

Notably, in the ideal deterministic LTI setting with condition \eqref{eq:GPE}, $\Dbc$ always has rank deficiency
    $\rank{\Dbc} = \rank{\Zb}, $
such that the image representation \eqref{eq:imageRep0} implies a unique (and exact) linear predictor mapping $\ybs = \hat\ybs(\xib, \ubs)$.
However, this rank deficiency, and with it the unique predictions, are typically lost in the presence of noise or nonlinearities. 
Specific for DPC, a commonly used remedy is given by the addition of a regularization term $h(\ab)$ to the objective function \eqref{eq:DPCcost} \cite{Coulson2019DeePC, Berberich2020stability}, which
is in
the focus of this work.

\subsection{Regularized DPC}

The optimal control problem (OCP) that is solved for DPC in every time step can be stated as 
\begin{subequations}
\label{eq:DPC}
\begin{align}
\min_{\ubs,\ybs,\ab} 
 J(\xib, \ubs, &\ybs) + h(\ab) \label{eq:DPCcost} \\
\text{s.t.} \quad \quad  \begin{pmatrix}
     \xib \\ \ubs \\ \ybs 
\end{pmatrix} &= \begin{pmatrix}
    \Wb \\ \Ub \\ \Yb 
\end{pmatrix}\ab, \label{eq:DPCeqConstr}\\
\left(\ubs, \ybs \right) &\in \Uc \times \Yc.  \label{eq:DPCsetConstr}
\end{align}
\end{subequations}
with control objective $J(\xib, \ubs, \ybs)$, regularization $h(\ab)$, and input-output constraints $\Uc \times \Yc$. Conditions for equivalence of DPC and MPC are well established for some special cases. In particular, the equivalence with MPC for LTI systems holds with exact data satisfying \eqref{eq:GPE} and setting the regularization to $h(\ab) = 0$ \cite{Coulson2019DeePC}. In the presence of noise or nonlinearities, the unregularized OCP may use the (unrealistic) additional degrees of freedom $\rank{\Dbc} - \rank{\Zb}>0$ to greedily minimize the objective function $J(\xib, \ubs, \ybs)$. Given enough data (columns), it is established and realistic (see, e.g., \cite{Breschi2022new, Mattsson2024, KLAEDTKE2023} for a discussion) to make the following assumption.
\begin{assum}\label{assum:fullRank}
    The data matrix $\Dbc$ has full row-rank.
\end{assum}
Note that full row-rank of $\Dbc$ renders \eqref{eq:DPCeqConstr} meaningless without regularization, since there is an $\ab$ solving \eqref{eq:DPCeqConstr} for any arbitrary left-hand side trajectory. 

\subsection{Trajectory-specific effect of regularization and implicit predictors}

We briefly introduce and discuss the two main analysis tools used in the paper. For a more extensive discussion and motivation, we refer to \cite{KLAEDTKE2023, Klaedtke2025_AT_Accepted}.
\begin{defn}[\cite{Klaedtke2025_AT_Accepted}]\label{def:trajectorySpecific}
    We call the solution $h^\ast(\wb)$ to 
\begin{equation}
    h^\ast(\wb) := h(\ab^\ast(\wb)) 
=  \min_{\ab} 
  h(\ab) \quad  \text{s.t.} 
 \quad  \wb = \Dbc\ab \label{eq:OP_trajectorySpecific}
\end{equation} 
  the \textit{trajectory-specific effect} of $h(\ab)$ given $\Dbc$.
\end{defn}
Note that the trajectory $\wb$ is not an optimization variable, making additional constraints like \eqref{eq:DPCsetConstr} irrelevant. Intuitively, the trajectory specific effect $h^\ast(\wb)$ represents a price tag, indicating the cost of synthesizing a specific trajectory $\wb$ according to the choice of $h(\ab)$. Additionally, \eqref{eq:OP_trajectorySpecific} naturally emerges as an inner optimization problem to \eqref{eq:DPC}, allowing for the elimination of the auxiliary variable $\ab$.
\begin{defn}[\cite{KLAEDTKE2023}]\label{def:implicit_predictor}
    We call $\hat\ybs(\xib,\ubs)$ an \textit{implicit predictor} for an OCP if including the constraint $\ybs = \hat\ybs(\xib, \ubs)$ does not alter the (set of) minimizers $(\ubs^\ast, \ybs^\ast)$ and the optimal value.    
\end{defn}

While DPC does not enforce a prediction model, the multistep predictor $\hat\ybs_\text{DPC}(\xib,\ubs)$ consistently aligns with its predictions, reflecting the predictive behavior DPC attributes to the data-generating system. Its accuracy in capturing the true system's properties depends on the design choices made in DPC.
As discussed in \cite{Klaedtke2025_AT_Accepted},  Definition~\ref{def:implicit_predictor} is more conceptional than constructive. However, a valid implicit predictor for the DPC problem \eqref{eq:DPC} can be constructed via
\begin{subequations}
\label{eq:implicitPredictorOP}
\begin{align}
\hat\ybs_\text{DPC}(\xib, \ubs) = \argmin_\ybs \min_{\ab} 
 J(\xib, \ubs, &\ybs) + h(\ab) \label{eq:implicitPredictorOPcost} \\
\text{s.t.} \quad \quad  \begin{pmatrix}
     \xib \\ \ubs \\ \ybs 
\end{pmatrix} &= \begin{pmatrix}
    \Wb \\ \Ub \\ \Yb 
\end{pmatrix}\ab, \label{eq:implicitPredictorOPeqConstr}\\
 \ybs  &\in \Yc.  \label{eq:implicitPredictorOPsetConstr}
\end{align}
\end{subequations}
Here, we treat $(\xib, \ubs)$ as parameters and optimize over $(\ab, \ybs)$, and hence additional set constraints $\ubs\in \Uc$ can be dropped while $\ybs \in \Yc$ still need to be considered. This construction procedure retrieves a valid implicit predictor, since \eqref{eq:implicitPredictorOP} naturally occurs as an inner optimization problem to \eqref{eq:DPC}. Thus, its parametric solution $\hat\ybs_\text{DPC}(\xib, \ubs)$ satisfies Definition~\ref{def:implicit_predictor}.
Given data $\Dbc$, \eqref{eq:implicitPredictorOP_1Norm} can often be solved simply using an appropriate multiparametric solver like MPT \cite{MPT3}, and we show examples of this in Fig.~\ref{fig:scalingCRs} and~\ref{fig:symmetryPredictor}. However, rather than focusing on a particular implicit predictor for a specific data set generated by a distinct system, we are more interested in the structural properties of the implicit predictor that are applicable to any (realistic) data set produced by any system. These general properties are mainly induced by design choices when setting up the OCP \eqref{eq:DPC}, such as the choice of 1-norm regularization. We believe that understanding the consequences of these choices is paramount to meaningful DPC design.

\subsection{Preliminaries on atomic norm}\label{sec:atomic_norm}

We briefly discuss preliminaries on the so-called atomic norm \cite{Chandrasekaran2012_atomicNorm}, which are independent of the DPC context and serve as a foundational aspect for our novel contribution. The relation to the trajectory-specific effect of 1-norm regularization will be established in Section~\ref{sec:traj_specific}. Consider a set $\Dfrak$, which is a collection of atoms, with many different examples of such atoms given in \cite{Chandrasekaran2012_atomicNorm}. Next, assume that the elements $\wb^{(i)}\in\Dfrak$ are the extreme points of the convex hull $\conv{\Dfrak}$, that $\Dfrak$ is centrally symmetric about the origin, i.e., 
$
\wb^{(i)}\in\Dfrak \iff -\wb^{(i)}\in\Dfrak,
$
and that the origin is the centroid of the convex hull $\conv{\Dfrak}$.
Then, the gauge function (or Minkowski functional) of $\conv{\Dfrak}$,
\begin{subequations}\label{eq:AtomicNorm_Def}
    \begin{align}
    \|\wb\|_\Dfrak &:= \inf_t t \quad \text{s.t.}\quad t > 0,\; \wb \in t \cdot \conv{\Dfrak} \label{eq:AtomicNorm_Def1}\\
    &\,\,= \inf_{\ab} \!\!\sum_{\wb^{(i)}\in\Dfrak} \!\! \ab_i \quad \text{s.t.}\quad \ab \geq \zerob,\, \wb = \!\! \sum_{\wb^{(i)}\in\Dfrak} \!\!\ab_i \wb^{(i)}, \label{eq:AtomicNorm_Def2}
    \end{align}
\end{subequations}
is a norm, which is known as the \textit{atomic norm}. Intuitively, the gauge of a set (here, $\mathrm{conv}(\Dfrak)$) is the factor $t^\ast$ by which it needs to be scaled for a point $\wb$ to become an extreme point. Examples for applications of the atomic norm, including system identification and control, can be found in \cite{Chandrasekaran2012_atomicNorm}.

\section{Data usage and predictive behavior} \label{sec:main}

Having laid out the fundamentals, we are ready to present our main results. More precisely, we first introduce a useful reformulation of the DPC problem in Section~\ref{sec:1normToLin}. Following this, Section~\ref{sec:data_usage} demonstrates that certain data is never utilized by DPC with 1-norm regularization. Next, Section~\ref{sec:traj_specific} explores the trajectory-specific effects of 1-norm regularization in DPC, providing geometric intuition through the atomic norm. Finally, in Section~\ref{sec:implicit_predictor}, we analyze the predictive behavior by establishing structural properties of the implicit predictor induced by 1-norm regularization. To emphasize our findings on induced structures and geometric interpretations, we visualize them using low-dimensional numerical toy examples. The code used to generate these figures is available on our GitHub repository\footnote{see \url{https://github.com/Control-and-Cyberphysical-Systems/data-usage-dpc-1-norm}}.

\subsection{Useful reformulation of the DPC problem}\label{sec:1normToLin}

In mathematical programming, it is common to linearly reformulate costs (or constraints) involving the 1-norm by introducing additional variables. We split $\ab = \ab_+ - \ab_-$ into a positive part $\ab_+\geq \zerob$ and negative part $-\ab_-\leq \zerob$ with
$
\ab_{\pm}:= \begin{pmatrix}
        \ab_+^\top & \ab_-^\top
    \end{pmatrix}^\top$ such that $\Dbc \ab = 
    \begin{pmatrix}
        \Dbc & -\Dbc
    \end{pmatrix}\ab_\pm.
$
We then rewrite the DPC problem \eqref{eq:DPC} with $h(\ab) = \lambda \|\ab\|_1$ as
\begin{subequations}
\label{eq:DPC_1normLinear}
\begin{align}
\min_{\ubs,\ybs,\ab_\pm} 
 J(\xib, \ubs, &\ybs) +\lambda \oneb^\top \ab_\pm\label{eq:DPC_1normLinearCost} \\
\text{s.t.} \quad \quad  \begin{pmatrix}
     \xib \\ \ubs \\ \ybs 
\end{pmatrix} &= \Dbc_\pm\ab_\pm, \;\; \ab_\pm \geq \zerob\label{eq:DPC_1normLinearConicConstr}\\
\left(\ubs, \ybs \right) &\in \Uc \times \Yc.  \label{eq:DPC_1normLinearSetConstr}
\end{align}
\end{subequations}
with $\Dbc_\pm := \begin{pmatrix}
        \Dbc & -\Dbc
    \end{pmatrix}$. The linear combination of data columns in $\Dbc$ is thus replaced by a conical combination of mirrored data in $\Dbc_\pm$. This reformulation is possible because the entries of the optimizer satisfy $\ab_{+,i}^\ast\, \ab_{-,i}^\ast = 0$, implying
 $$
 \|\ab^\ast\|_1 = \|\ab_+^\ast - \ab_-^\ast\|_1 = \|\ab_+^\ast\|_1 + \|\ab_-^\ast\|_1 = \|\ab_\pm^\ast\|_1 = \oneb^\top \ab_\pm^\ast.
 $$ 
We note that this particular reformulation is done to develop the following theory. When considering solution speed, other reformulations like in \cite[Sec.~6.1.1]{boyd2004} may perform better due to different sparsity patterns.

\subsection{Unused data in DPC with 1-norm regularization}\label{sec:data_usage}
We denote with
    $\Dfrak := \{\wb^{(i)}\in\R^{(m+p)(N_p+N)} | i\in {1, ..., \ell}\}$
the set of all I/O data trajectories $\wb^{(i)}$ contained as columns in $\Dbc$, and likewise for $\wb^{(i)}_{\pm} \in \Dfrak_\pm$ being columns of $\Dbc_\pm$. In this notation, the trajectory synthesis via linear \eqref{eq:DPCeqConstr} or conical \eqref{eq:DPC_1normLinearConicConstr} combinations can be equivalently written as
\begin{align*}
    &\wb =  \!\!\sum_{\wb^{(i)} \in \Dfrak} \ab_i \wb^{(i)} 
    \! 
    \!\!\!\!\!\!\!\! &&\iff \! \wb = \!\! \sum_{\wb^{(i)}_{\pm} \in \Dfrak_\pm} \! \ab_{\pm,i} \wb^{(i)}_\pm,  \;  \ab_{\pm,i} \geq 0 \\
    \iff& \wb \! \in \mathrm{span}(\Dfrak) \!\!\!\!\! &&\iff  \wb \in \mathrm{coni}(\Dfrak_\pm).
\end{align*}
The linear span $\mathrm{span}(\cdot)$ refers to the set of all trajectories that can be synthesized through linear combinations, while the conical hull $\mathrm{coni}(\cdot)$ refers to those that can be synthesized through conical combinations.
Now, consider the subset $\overline{\Dfrak}_\pm \subseteq \Dfrak_\pm$ of extreme points of the convex hull $\conv{\mathfrak{D}_\pm}$, for which standard construction algorithms, such as Quickhull \cite{Barber1996_quickhull}, are available.
In the same way as $\Dfrak_\pm$ is constructed from $\Dfrak$ by mirroring its elements, we define a set $\overline{\Dfrak}:=\Dfrak \cap \overline{\Dfrak}_\pm$ containing only the ``un-mirrored'' data points of $\overline{\Dfrak}_\pm$. 
The relations between these sets is visualized in Fig.~\ref{fig:visualizeDataSets}.
\begin{figure}
	\centering
	\def\svgwidth{\linewidth}
    \includeinkscape{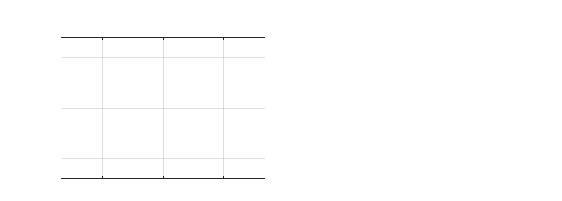}
	\caption{Visualization of data sets with elements in $\Dfrak$ drawn from $\Nc(\zerob, \Ib)$. (a)~Original sets $\Dfrak$ in green ($\times$), and $\Dfrak_\pm$ in green and orange ($\times, \circ$). The convex hull $\conv{\Dfrak_\pm}$ is also highlighted. (b) Reduced sets $\overline\Dfrak$ in green ($\times$), and $\overline\Dfrak_\pm$ in green and orange ($\times, \circ$).}
	\label{fig:visualizeDataSets}
\end{figure}
Next, we show that the linear spans and conical hulls remain unchanged.
\begin{lem}\label{lem:conicalHull}
    Consider the trajectory data set $\Dfrak$, the mirrored set $\Dfrak_\pm$, and their subsets $\overline\Dfrak, \overline\Dfrak_\pm$. 
    Their conical hulls $\mathrm{coni}(\cdot)$ and linear spans $\mathrm{span}(\cdot)$ satisfy
    $
        \mathrm{span}(\overline\Dfrak) = \mathrm{coni}(\overline\Dfrak_\pm)  =\mathrm{coni}(\Dfrak_\pm) = \mathrm{span}(\Dfrak).
    $
\end{lem}
\begin{proof}
    The relations $\mathrm{coni}(\Dfrak_\pm) = \mathrm{span}(\Dfrak)$ and $\mathrm{coni}(\overline\Dfrak_\pm) = \mathrm{span}(\overline\Dfrak)$ follow trivially from construction of the mirrored data sets $\Dfrak_\pm, \overline\Dfrak_\pm$. Furthermore, we have $\mathrm{coni}(\overline\Dfrak_\pm)=\mathrm{coni}(\Dfrak_\pm)$, since $\Dfrak_\pm\setminus\overline\Dfrak_\pm \subseteq \mathrm{coni}(\overline\Dfrak_\pm)$ as per the construction procedure of $\overline\Dfrak_\pm$.
\end{proof}
 The following theorem establishes that, when using 1-norm regularization, any point removed during the construction of $\overline\Dfrak$ (or $\overline\Dfrak_\pm$) is completely irrelevant for the DPC problem.
\begin{thm}\label{thm:Preprocessing}
    Consider the DPC problem \eqref{eq:DPC} with 1-norm regularization $h(\ab) = \lambda \|\ab\|_1$. Replacing the data matrix $\Dbc$ with a data matrix $\overline\Dbc$ whose columns are the elements of $\overline\Dfrak$ does not alter the optimal cost and optimizers $(\ubs^\ast, \ybs^\ast)$.
\end{thm}
\begin{proof}
   We prove the claim by demonstrating that none of the data trajectories in $\Dfrak\setminus \overline\Dfrak$ are used, even when considering the entire dataset $\Dfrak$. That is, we have $\ab_j^\ast=0$ for all $j$ with $\wb^{(j)}\in \Dfrak\setminus \overline\Dfrak$. Hence, their removal does not affect the outcome of the OCP in the claimed sense. We shift our focus to the equivalent reformulation \eqref{eq:DPC_1normLinear}, and proceed with an indirect proof. Consider an optimizer triple $(\ubs^\ast, \ybs^\ast, \ab_\pm^\ast)$ with corresponding optimum
    $$
        V^\ast = J(\xib, \ubs^\ast, \ybs^\ast) + \lambda \oneb^\top \ab_\pm^\ast = J(\xib, \ubs^\ast, \ybs^\ast) + \lambda \sum_{\wb^{(i)}_{\pm} \in \Dfrak_\pm} \! \ab_{\pm,i}^\ast
    $$
    that has one (or more) $\ab_{\pm, j}^\ast \neq 0$ with $\wb^{(j)}_\pm\in\Dfrak_\pm\setminus \overline\Dfrak_\pm$. However, by construction of $\overline\Dfrak_\pm$, there exists a conical combination 
    $$
        \wb^{(j)}_\pm = \sum_{\wb^{(i)}_{\pm} \in \overline\Dfrak_\pm} \!\! \overline\ab_{\pm,i} \wb^{(i)}_\pm, \quad \overline\ab_{\pm,i} \geq 0 \;\; \text{with} \!  \sum_{\wb^{(i)}_{\pm} \in \overline\Dfrak_\pm}  \overline\ab_{\pm,i} < 1.
    $$
     Now, in the trajectory synthesis \eqref{eq:DPC_1normLinearConicConstr}, consider replacing $\wb^{(j)}_\pm$ with this new conical combination as follows
     \begin{align*}
         \begin{pmatrix}
            \xib \\
            \ubs^\ast \\
            \ybs^\ast
        \end{pmatrix}
        &=
        \sum_{\wb^{(i)}_{\pm} \in \Dfrak_\pm} \!\! \ab_{\pm,i}^\ast \wb^{(i)}_\pm 
        = \ab_{\pm,j}^\ast \wb^{(j)}_\pm + \sum_{\wb^{(i)}_{\pm} \in \Dfrak_\pm \setminus \wb^{(j)}_\pm} \!\! \ab_{\pm,i}^\ast \wb^{(i)}_\pm \\
        &= \ab_{\pm,j}^\ast \sum_{\wb^{(i)}_{\pm} \in \overline\Dfrak_\pm} \!\! \overline\ab_{\pm,i} \wb^{(i)}_\pm + \sum_{\wb^{(i)}_{\pm} \in \Dfrak_\pm \setminus \wb^{(j)}_\pm} \!\! \ab_{\pm,i}^\ast \wb^{(i)}_\pm.
     \end{align*}
     The associated optimizer satisfies all constraints and yields
     \begin{align*}
         \overline V &= J(\xib, \ubs^\ast, \ybs^\ast) +\lambda\ab_{\pm,j}^\ast \underbrace{\sum_{\wb^{(i)}_{\pm} \in \overline\Dfrak_\pm} \!\! \overline\ab_{\pm,i}}_{<1} + \lambda\sum_{\wb^{(i)}_{\pm} \in \Dfrak_\pm \setminus \wb^{(j)}_\pm} \!\! \ab_{\pm,i}^\ast \\
         &< J(\xib, \ubs^\ast, \ybs^\ast) + \lambda\ab_{\pm,j}^\ast  + \lambda\sum_{\wb^{(i)}_{\pm} \in \Dfrak_\pm \setminus \wb^{(j)}_\pm} \!\! \ab_{\pm,i}^\ast = V^\ast.
     \end{align*}
     Hence, $\ab_\pm^\ast$ is not optimal, which proves the claim.
\end{proof}
Therefore, data in $\Dfrak\setminus\overline\Dfrak$ can be discarded without affecting the outcome of the OCP. For instance, the example in Fig.~\ref{fig:visualizeDataSets} retains only $3$ of the original $8$ data samples. The implications for DPC are manifold, including, among others:\\ 
    1)  We recommend constructing the relevant data set $\overline\Dfrak$ offline as a preprocessing step to enhance online computation speed. \\
    2) The criterion for including or excluding trajectory data is implicitly (and perhaps unknowingly) determined by the choice of regularization and may or may not align with the characteristics of the true system. Similar aspects will be explored in Section~\ref{sec:implicit_predictor}. \\
    3) The choice of coordinates for the recorded data has a big impact on the shape of $\mathrm{conv}(\Dfrak_\pm)$, and there may be detrimental choices that discard data samples which are important for the control task at hand. \\
    4) The exploration strategy used for generating data might also have a big impact on the shape of $\mathrm{conv}(\Dfrak_\pm)$ and share of unused trajectory samples $\Dfrak\setminus\overline\Dfrak$. \\
However, exploring these aspects in depth is beyond the scope of this paper and will be the focus of future work.

\subsection{Trajectory-specific effect of 1-norm regularization}\label{sec:traj_specific}

The following theorem establishes a connection between the trajectory-specific effect of 1-norm regularization and the atomic norm introduced in Section~\ref{sec:atomic_norm}. 
\begin{thm}\label{thm:1NormTrajSpecificAtomic}
    Under Assumption~\ref{assum:fullRank}, the trajectory-specific effect of the 1-norm regularization $h(\ab) =  \lambda\|\ab\|_1$ is given by the (scaled) atomic norm
    $
        h^\ast(\wb) = \lambda\|\wb\|_{\overline\Dfrak_\pm}
    $
    with respect to the preprocessed and mirrored data set $\overline\Dfrak_\pm$.
\end{thm}
\begin{proof}
    We state the result for $\lambda = 1$ without loss of generality.
    First, note that the set $\overline\Dfrak_\pm$ satisfies all assumptions to induce an atomic norm $\|\wb\|_{\overline\Dfrak_\pm}$. 
    As in \eqref{eq:DPC_1normLinear}, we can rewrite the optimization problem \eqref{eq:OP_trajectorySpecific} with $h(\ab) = \|\ab\|_1$ as
\begin{equation}\label{eq:trajSpec_1normLinear}
    h^\ast(\wb) = \min_{\ab_\pm} 
      \oneb^\top \ab_\pm \quad
    \text{s.t.} \quad \wb = \Dbc_\pm\ab_\pm, \; 
    \ab_\pm \geq \zerob. 
    \end{equation}
    From Theorem~\ref{thm:Preprocessing}, we can replace $\Dbc_\pm$ by $\overline\Dbc_\pm$ without affecting the optimal solution. Furthermore, we equivalently rewrite \eqref{eq:trajSpec_1normLinear} with respect to the set $\overline\Dfrak_\pm$ as
     $$   h^\ast(\wb) = 
      \min_{\ab_\pm} \!\! \sum_{\wb^{(i)}_\pm\in\overline\Dfrak_\pm}\!\! \ab_{\pm,i} \;\; 
      \text{s.t.}\;\; \wb = \!\!\! \sum_{\wb^{(i)}_\pm\in\overline\Dfrak_\pm} \!\!\! \ab_{\pm,i} \wb^{(i)}_\pm, \; 
       \ab_\pm \geq \zerob$$ 
    This matches \eqref{eq:AtomicNorm_Def2} if the infimum is attained, which is indeed the case for all $\wb\in\R^{(m+p)(N_p+N)}$ due to Lemma~\ref{lem:conicalHull} with Assumption~\ref{assum:fullRank}.  
\end{proof}
This relation enables a geometric intuition of regularization costs in terms of $\mathrm{conv}(\overline\Dfrak_\pm)$ via its gauge (see our explanation below \eqref{eq:AtomicNorm_Def}). 
Furthermore, under Assumption~\ref{assum:fullRank}, Theorem~\ref{thm:1NormTrajSpecificAtomic} enables the reformulation of DPC \eqref{eq:DPC} with 1-norm regularization as atomic norm regularization in trajectory space
$$
    \min_{\wb} J(\wb) + \lambda  \|\wb\|_{\overline\Dfrak_\pm} \; \text{s.t.} \; \wb \in \Wbc,
$$
which will be further explored in future work.

\subsection{Predictive behavior of DPC with 1-norm regularization}\label{sec:implicit_predictor}

In this section, we analyze the predictive behavior of DPC with 1-norm regularization via the concept of implicit predictors (see Definition~\ref{def:implicit_predictor}), and provide two structural results with immediate practical implications. We assume that the control objective is a quadratic output-tracking formulation
\begin{equation}\label{eq:outputTrackingObj}
    J(\xib, \ubs, \ybs) = \|\ybs-\ybs_\text{ref}\|_\Qbc^2 + J_\ubs(\xib, \ubs) = \|\ybs\|_\Qbc^2 + J_\ubs(\xib, \ubs)
\end{equation}
with reference $\ybs_\text{ref}=\zerob$, positive semidefinite weighing matrix $\Qbc$, and arbitrary input control objective $J_\ubs(\xib, \ubs)$. 
Non-zero reference tracking $\ybs_\text{ref}\neq\zerob$ will be considered in future work, but can already be made compatible via an appropriate coordinate shift for the output (including data).
Similarly for simplicity, we assume that no additional output constraints are present, i.e., $\Yc = \R^{p N}$, and refer to \cite[Sec.~III.C]{KLAEDTKE2023} for general observations on their impact. 
Given these specifications and the reformulation in Section~\ref{sec:1normToLin} combined with Theorem~\ref{thm:Preprocessing}, \eqref{eq:implicitPredictorOP} can be simplified to 

\begin{align}\label{eq:implicitPredictorOP_1Norm}
\hat\ybs_\text{DPC}(\zb) = \argmin_\ybs \min_{\overline\ab_\pm} &
 \|\ybs\|_\Qbc^2 + \lambda \oneb^\top \overline\ab_\pm  \\
\text{s.t.} \quad \quad  \begin{pmatrix}
    \zb \\ \ybs 
\end{pmatrix} &= \overline\Dbc_\pm\overline\ab_\pm,\quad \overline\ab_\pm\geq \zerob.\nonumber
\end{align}
We can immediately see that \eqref{eq:implicitPredictorOP_1Norm} is a multiparametric quadratic program (mpQP) with parameters $\zb := \begin{pmatrix}
    \xib^\top & \ubs^\top
\end{pmatrix}^\top$. It is well known \cite{borrelli2017} that the optimal solution to such an mpQP is a piecewise affine (PWA) function
$$
    \hat\ybs_\text{DPC}(\zb) = \Fb_i\zb+\gb_i, \;\zb \in\Cbc_i
$$ 
over polyhedral sections $\Cbc_i$ of the parameter space (here, the state-input space $ \R^{(m+p)N_p+mN}=\bigcup\Cbc_i$). Furthermore, the polyhedral sections are so-called critical regions (CRs) characterized by (in)activity of inequality constraints. Note that inactivity of an inequality constraint in \eqref{eq:implicitPredictorOP_1Norm}, i.e., $\overline\ab_{\pm, i}>0$, is equivalent to a corresponding data column $\overline\wb_{\pm}^{(i)}$ in $\overline\Dbc_\pm$ (or, equivalently, atom in $\overline\Dfrak_\pm$) being used for trajectory synthesis. Consequently, each CR can be linked to a particular subset of data employed for trajectory synthesis. However, it is not immediately clear which data columns are preferred for synthesis in specific regions of the state-input space. While an exhaustive analysis of the associated CRs as well as their dependence on the data $\Dbc$, and weights $\Qbc, \lambda$ is beyond this paper's scope, the following result provides a piece of the puzzle by establishing a scaling property of the CRs with $\lambda$.
\begin{thm}\label{thm:scaling}
    Consider the two implicit predictors $\hat\ybs_\text{DPC}(\zb)$ resulting from \eqref{eq:implicitPredictorOP_1Norm} and $\hat\ybs_{\text{DPC}, \eta}(\zb)$ resulting from a scaled regularization weight $\tilde\lambda = \eta \cdot \lambda$. Their respective CRs $\Cbc_i$ and $\Cbc_{\eta,i}$ are defined by the same set of (in)active constraints and satisfy the scaling relation $\Cbc_{\eta,i} = \eta \cdot \Cbc_i$.
\end{thm}
\begin{proof}
    For brevity, we rewrite \eqref{eq:implicitPredictorOP_1Norm} as
    \begin{align*}
        \min_\sbb \sbb^\top \Hb \sbb + \lambda \fb^\top \sbb \quad\text{s.t.}\quad \Ab \sbb \leq \zerob, \; \Gb\sbb = \Eb\zb 
    \end{align*}
    with $\Hb := \blkdiag{\Qb, \zerob}$, $\fb^\top:=\begin{pmatrix}\zerob^\top & \oneb^\top\end{pmatrix}$, $\Ab := 
        \begin{pmatrix}
            \zerob & -\Ib
        \end{pmatrix}$,
    \begin{align*}
        \Gb := \begin{pmatrix} \begin{matrix}
            \zerob \\ -\Ib
        \end{matrix}
             & \overline\Dbc_\pm
         \end{pmatrix}
        , \quad  \Eb := \begin{pmatrix}
            \Ib \\
            \zerob
        \end{pmatrix},\;\text{and}\; \sbb:=\begin{pmatrix}
        \ybs \\ \overline\ab_\pm
    \end{pmatrix}.
    \end{align*}
    Next, we eliminate the equality constraint by re-parametrizing $\sbb = \Gb^+\Eb\zb+\Vb\alphab$, where $\Vb$ is a matrix whose columns form a basis of the nullspace $\kernel{\Gb}$. Rewriting the problem in terms of the new optimization variable $\alphab$ and dropping cost terms constant w.r.t. $\alphab$ yields
    \begin{align}
        \min_\alphab \alphab^\top \Vb^\top &\Hb \Vb^\top \alphab + \left(\zb^\top \Eb^\top\Gb^{+, \top}\Hb + \lambda \fb^\top \right)\Vb\alphab \label{eq:scalingProof_OPorig}\\&\text{s.t.}\quad \Ab \Gb^+\Eb\zb + \Ab \Vb\alphab \leq \zerob. \nonumber
    \end{align}
    Now, consider a similar procedure with the scaled weight $\tilde\lambda = \eta\cdot\lambda$. However, the equality constraints are eliminated using a scaled matrix $\widetilde\Vb = \eta\Vb$, resulting in the parametrization $\sbb = \Gb^+\Eb\zb+\eta\, \Vb\betab$. Furthermore, we introduce the scaled state-input coordinates $\tilde\zb$ with $ \tilde\zb := \eta^{-1} \zb$. Reformulating the optimization problem in terms of these modifications yields
    \begin{align}
        \min_{\betab} \eta^2\betab^\top \Vb^\top &\Hb \Vb^\top \betab + \left(\eta\,\tilde\zb^\top \Eb^\top\Gb^{+, \top}\Hb + \eta\,\lambda \fb^\top \right)\eta\,\Vb\betab \nonumber\\ &\text{s.t.}\quad \eta \, \Ab \Gb^+\Eb\tilde\zb  + \eta\, \Ab \Vb\betab \leq \zerob \nonumber 
    \end{align}
    Since neither the scaling of the entire objective function with $\eta^2$ nor the l.h.s of the inequality constraints with $\eta$ affects the solution, we end up with
    \begin{align}
        \min_{\betab} \betab^\top \Vb^\top &\Hb \Vb^\top \betab + \left(\tilde\zb^\top \Eb^\top\Gb^{+, \top}\Hb + \lambda \fb^\top \right)\Vb\betab \label{eq:scalingProof_OPscaled}\\ &\text{s.t.}\quad  \Ab \Gb^+\Eb\tilde\zb  +  \Ab \Vb\betab \leq \zerob \nonumber.
    \end{align}
    Comparing \eqref{eq:scalingProof_OPorig} and \eqref{eq:scalingProof_OPscaled}, we see that their parametric optimizers satisfy $\alpha^\ast(\zb) = \beta^\ast(\tilde\zb)$. 
    Hence, both share the same partitioning $\widetilde\Cbc_{\eta,i} = \Cbc_i$ of the (scaled) state-input space, and reverting this scaling w.r.t. the CRs yields 
    $
    \tilde\zb \in \widetilde\Cbc_{\eta, i} 
    \!\!\iff \!\!\zb \in \eta \, \widetilde\Cbc_{\eta, i}
    $ such that $\Cbc_{\eta, i} = \eta \, \widetilde\Cbc_{\eta, i} = \eta \,\Cbc_i$. Finally, neither of the re-parametrizations changed the order or meaning of the inequality constraints. Hence, all pairs $(\tilde \Cbc_i, \Cbc_i)$ share the same set of (in)active equality constraints.
\end{proof}
The scaling property established in Theorem~\ref{thm:scaling} is visualized in Fig.~\ref{fig:scalingCRs}. Importantly, since every CR is associated with a subset of data being used for trajectory synthesis, Theorem~\ref{thm:scaling} characterizes how data usage changes throughout the state-input space with changing regularization weight $\lambda$. Hence, tuning of $\lambda$ should ideally be informed by this property.
\begin{figure}
	\centering
	\def\svgwidth{\linewidth}
    \includeinkscape{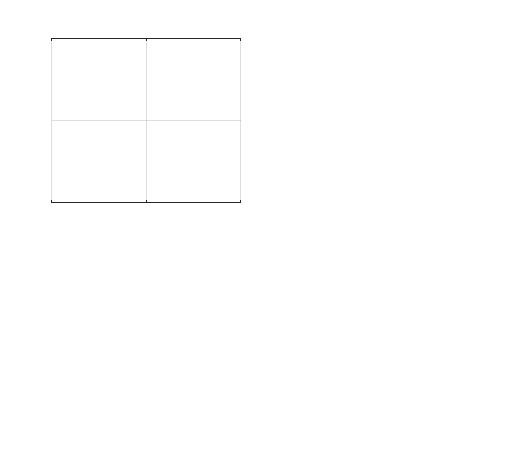}
	\caption{CRs of the implicit predictor $\hat\xbs_\text{DPC}(\xb_0, \ubs)$ for $\Qbc=1$ and the same data $\overline\Dfrak$ (blue) as in Fig.~\ref{fig:symmetryPredictor}. Scaling of the regularization weight (a) $\lambda=100$ (green) and (b) $\lambda=50$ (orange) results in an equivalent scaling of the CRs (compare bottom left and top right). Note that top and bottom row only differ in their zoom level.}
	\label{fig:scalingCRs}
\end{figure}
The next result concerns a symmetry that is implied solely by the design choices made in the setup of the OCP.
\begin{thm}\label{thm:symmetry}
    Consider the DPC problem \eqref{eq:DPC} with $J(\zb, \ybs) = J_\ybs(\zb, \ybs)+J_u(\zb)$ and $\Yc=\R^{p N}$. If $J_\ybs(\zb, \ybs)$ and $h(\ab)$ are even functions, and the parametric optimizer $\ybs^\ast(\zb)$ is unique for all feasible $\zb$, then the implicit predictor $\hat\ybs_\text{DPC}(\zb)$ resulting from \eqref{eq:implicitPredictorOP} is an odd function.
\end{thm}
\begin{proof}
    Given the setup, we simplify \eqref{eq:implicitPredictorOP} by dropping \eqref{eq:implicitPredictorOPsetConstr} and $J_u(\zb)$, since the latter is unrelated to $\ybs$, yielding
$$
 \argmin_\ybs \min_{\ab} 
 J_\ybs(\zb, \ybs) + h(\ab) \quad
\text{s.t.} \quad \begin{pmatrix}
     \zb \\ \ybs 
\end{pmatrix} = \begin{pmatrix}
    \Zb \\ \Yb 
\end{pmatrix}\ab.
$$
Due to our symmetry assumptions on $J_\ybs(\zb, \ybs)$ and $h(\ab)$, this problem is equivalent to 
\begin{align*}
    -\hat\ybs_\text{DPC}(-\zb) &= \argmin_{\ybs} \min_{\ab} 
     J_\ybs(-\zb, -\ybs) + h(-\ab) \\
    &\qquad \qquad \quad\text{s.t.} \quad \begin{pmatrix}
         -\zb \\ -\ybs 
    \end{pmatrix} = \begin{pmatrix}
        \Zb \\ \Yb 
    \end{pmatrix}(-\ab).
\end{align*}
Hence, we have 
    $-\hat\ybs_\text{DPC}(-\zb)=\hat\ybs_\text{DPC}(\zb)$, which is odd.
\end{proof}
Our setup with $J_\ybs(\xib,\ubs, \ybs) = \|\ybs\|_\Qbc^2$ and $h(\ab)=\lambda\|\ab\|_1$ clearly satisfies the symmetry assumptions of Theorem~\ref{thm:symmetry}. 
The practical implications become evident when comparing this structure to that of potential data-generating systems. That is, if the true system dynamics share this symmetry, which is trivially true for LTI systems, the implicit predictor might be able to represent the system well throughout the whole trajectory space. For nonlinear systems without this symmetry, the desired approximation of the true dynamics may only ever hold locally. Even worse, exploration of certain parts of the trajectory space may actively degrade DPC's prediction accuracy (and thus usually its performance) in other parts, where the mirrored data columns are being used for trajectory synthesis. These considerations are visualized in Fig.~\ref{fig:symmetryPredictor}. 
\begin{figure}
	\centering
	\def\svgwidth{\linewidth}
    \includeinkscape{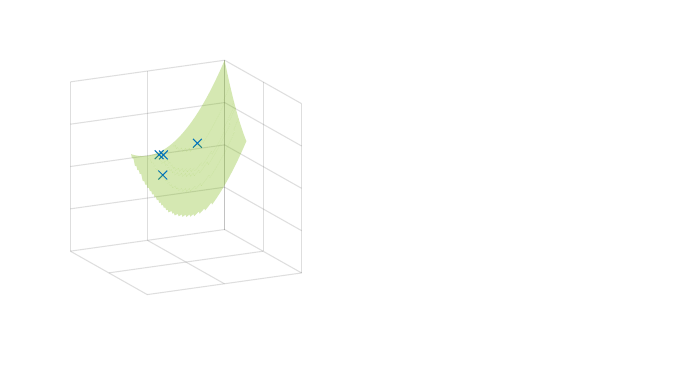}
	\caption{Comparison of true dynamics and implicit predictor. (a) System dynamics $f(\xb, \ub)= 2 \xb^2+2\ub^2-1$ (green) with $20$ data samples $\Dfrak$ generated from $\xb_0^{(i)}, \ubs^{(i)}\sim\Uc_{[-1,1]}$ (blue). (b) Implicit predictor $\hat\xbs_\text{DPC}(\xb_0, \ubs)$ (orange) for $\lambda = 100, \Qbc = 1$ resulting from the $8$ relevant data samples $\overline\Dfrak$ (blue).}
	\label{fig:symmetryPredictor}
\end{figure}
Clearly, the true dynamics are not well represented throughout the whole trajectory space. In particular, Theorem~\ref{thm:symmetry} always implies $\hat\xbs_\text{DPC}(\zerob, \zerob)=\zerob$, leading to systematic prediction errors near the origin if $f(\zerob, \zerob)\neq\zerob$. We recommend either consciously leveraging this symmetry property when suitable for the underlying system or modifying the OCP \eqref{eq:DPC} to break it in an appropriate sense, e.g., by using affine or conical combinations \cite{Padoan2023}. These approaches may be explored in future work.

\section{Conclusions and Outlook}
\label{sec:Conclusions}

We analyzed the data usage and predictive behavior of DPC with 1-norm regularization using the concepts of trajectory-specific effect of regularization (Definition~\ref{def:trajectorySpecific}) and implicit predictors (Definition~\ref{def:implicit_predictor}). Our findings reveal that certain trajectory samples are never used for predictions, allowing their removal in a pre-processing step, while raising concerns about how this data usage may affect control performance. Additionally, we demonstrated that the trajectory-specific effect of 1-norm regularization is equivalent to the atomic norm \cite{Chandrasekaran2012_atomicNorm} for this reduced (and mirrored) dataset $\overline\Dfrak_\pm$, enabling intuitive geometric interpretations via the gauge of the convex hull $\mathrm{conv}(\overline\Dfrak_\pm)$. Regarding the predictive behavior, we showed that the implicit predictor $\hat\ybs_\text{DPC}(\xib, \ubs)$ is a PWA function whose CRs $\Cbc_i$ dictate local data usage and scale linearly with the regularization parameter $\lambda$. We also established a symmetry property $\hat\ybs_\text{DPC}(\xib, \ubs)=-\hat\ybs_\text{DPC}(-\xib, -\ubs)$, which can enhance or detract from prediction (and potentially control) performance, depending on whether the true system exhibits this symmetry.

In future work, we will further investigate the link to the atomic norm, conduct an exhaustive analysis of the implicit predictor, and extend our results to more general scenarios, such as general constraints $\Yc$, broader control objectives than \eqref{eq:outputTrackingObj}, or a mix of regularizations.

\end{document}